\documentclass[a4paper,11pt]{article}
\usepackage[T1]{fontenc}
\usepackage[utf8]{inputenc}
\usepackage{lmodern}
\usepackage{amsfonts}
\usepackage{amssymb}
\usepackage{amsmath}
\usepackage{fullpage}
\usepackage{amsthm}
\usepackage[us,hhmmss]{datetime}
\usepackage{cite}
\usepackage{authblk}
\usepackage{graphicx}
\usepackage{url}

\numberwithin{equation}{section}	

\title{Polaron models with regular interactions at strong coupling}
\author{Krzysztof My\'{s}liwy\footnote{krzysztof.mysliwy@ist.ac.at}\ \ and Robert Seiringer\footnote{robert.seiringer@ist.ac.at}} 
\affil{\textit{IST Austria, Am Campus 1, 3400 Klosterneuburg, Austria}}

\newtheorem{thm}{Theorem}

\newtheorem{lemma}[thm]{Lemma}

\newtheorem{conjecture}[thm]{Conjecture}

\theoremstyle{remark}
\newtheorem{rem}{Remark}[thm]

\newcommand{\comment}[1]{}
\newcommand{\mpek}{M^{\mathrm{Pek}}}
\newcommand{\dd}{d} 
\newcommand{\GG}{G_{\psi,\varphi}}
\theoremstyle{definition}

\begin{document}
\maketitle
\begin{abstract}
We study a class of polaron-type Hamiltonians with sufficiently regular form factor in the interaction term. We investigate the strong-coupling limit of the model, and prove suitable bounds on the ground state energy as a function of the total momentum of the system. These bounds agree with the semiclassical approximation to leading order. The latter corresponds here to the situation when the particle undergoes harmonic motion in a potential well whose  frequency is  determined by the corresponding Pekar functional. We show that for all such models the effective mass diverges in the strong coupling limit, in all spatial dimensions. Moreover, for the case when the phonon dispersion relation grows at least linearly with momentum, the  bounds result in an asymptotic formula for the effective mass quotient, a quantity generalizing the usual notion of the effective mass. This asymptotic form agrees with the semiclassical Landau--Pekar formula and can be regarded as the first rigorous confirmation, in a slightly weaker sense than usually considered, of the validity of the semiclassical formula for the effective mass.
\end{abstract}
\section{Introduction and main results}

\subsection{The model}The polaron problem concerns the motion of a quantum particle of mass $m$ exchanging energy and momentum with a large environment modeled by a  bosonic field. The model has a long history tracing back to the thirties \cite{Lan1,Fro1,Pek,Fey1} but due to its basic character it remains a model of reference in many problems, and is still under active investigation in condensed matter physics; we refer to \cite{Dev, Grusdt} for an overview and further references. The models under study here are defined by the Hamiltonian
\begin{equation}\label{ham}
\mathbb{H}=  -\frac{1}{2m}\Delta_x+ \int_{\mathbb{R}^d} \epsilon(k)a^{\dagger}_k a_k dk +\sqrt{\alpha} \int_{\mathbb{R}^d} \left (v(k) a_k e^{ik\cdot x} +\overline{v(k)}a^{\dagger}_k e^{-ik\cdot x}\right) dk.
\end{equation} This operator acts on $L^2(\mathbb{R}^d)\otimes \mathcal{F}$ with $\mathcal{F}$ the bosonic Fock space over $L^2(\mathbb{R}^d)$, and with $a_k, a^{\dagger}_k$ the usual annihilation and creation operators.  The phonon {\em dispersion relation} $\epsilon$ is a positive function, $v$ quantifies the interaction of the particle with the  field modes and is referred to as the \emph{form factor}, and $\alpha$ is a positive coupling constant, traditionally appearing in \eqref{ham} under the square root. We assume that  $\inf_{k\in\mathbb{R}^d} \epsilon(k) >0$ and $v \in L^2(\mathbb{R}^d)$, in which case \eqref{ham} is well-defined as a self-adjoint operator on the intersection of the domains of $\Delta_x$ and the field energy  $\int \epsilon(k) a^\dagger_k a_k dk$, respectively.  Moreover, 
 we can then readily define two functions naturally related to this Hamiltonian: the \emph{Pekar kernel}
\begin{equation}\label{def:h}
h(x):=\frac{1}{(2\pi)^{d/2}}\int_{\mathbb{R}^d}  \frac{v(k)}{\sqrt{\epsilon(k)}}e^{ik\cdot x}dk
\end{equation} and the position space potential 
\begin{equation}
\eta(x):=\frac{1}{(2\pi)^{d/2}}\int_{\mathbb{R}^d}  v(k) e^{ik\cdot x}dk.
\end{equation}   
We shall impose further regularity assumptions on $h$ and $\eta$, namely that $h$ is in the Sobolev space $W^{2,2}(\mathbb{R}^d)$ and that $\eta$ is in $W^{1,2}(\mathbb{R}^d)$. Equivalently, the functions $k\mapsto v(k)(1+k^2)\epsilon(k)^{-1/2}$ and $k\mapsto v(k) (1+k^2)^{1/2}$ are in $L^2(\mathbb{R}^d)$. 
For simplicity, we shall also assume that  the form factor and the dispersion relation depend on $|k|$ only, and that the latter is a continuous function of $|k|$. If all these conditions are satisfied, we call $\mathbb{H}$  \emph{regular}. 

Our main interest lies in the strong-coupling limit of very large $\alpha$, and its connection to the semiclassical limit described below. This problem has been studied in the mathematical physics literature \cite{LT,DV,FrankSe} in the special case of the {\em Fr\"ohlich model} corresponding to $d=3$,  $v(k)=(\sqrt{2}\pi|k|)^{-1}$ and $\epsilon(k)=1$ in appropriate units. It corresponds to the original polaron problem 
addressing the  important problem of electronic conductivity in ionic crystals.
Our goal here is to analyze the strong-coupling limit in the regular case, where, on the one hand, one does need to worry about the UV divergences as in the Fr\"ohlich model, but at the same time the useful scaling properties found therein are lost.  We believe that performing the strong-coupling analysis for polaron models other than the original Fr\"ohlich Hamiltonian may be of relevance as various versions of the polaron problem, with more general choices of the form factor and the dispersion relation, are being considered in the literature, mostly in the context of the physics of cold atoms, e.g. in the Bose polaron model and its analog, the angulon model \cite{Grusdt, Tempere, MiszaEnderalp, Timour, PenaArdila}. The rigorous results obtained, even if proved for simplified versions of the problem, may be practically useful e.g. as a reference point for numerical calculations. At the same time, the regularity enables us to prove new results concerning the validity of the semiclassical approximation to the effective mass, which constitutes an outstanding open problem. Our result on the effective mass is applicable in the case of a dispersion relation growing at least linearly in $|k|$ as in the case of the Bose polaron, thus excluding the Fr\"ohlich polaron, although we expect that our methods can serve as a starting point in future investigations on this problem also in this case. 

\subsection{Basic considerations and definitions}

Because of translation invariance, the  Hamiltonian \eqref{ham} commutes with the total momentum 
\begin{equation}
-i \nabla_x +\underbrace{\int k\, a_k^{\dagger}a_k \, dk}_{=:P_f}
\end{equation} and it can be cast, using a transformation due to Lee, Low and Pines \cite{LLP}, in the unitarily equivalent form 

\begin{equation}
\frac{1}{2m}\left(-i \nabla_x -P_f\right)^2+\mathbb{F}+\sqrt{\alpha}\mathbb{V}
\end{equation}
where $\mathbb{F}=\int \epsilon(k)a^{\dagger}_k a_k dk$ and $\mathbb{V}= \int  (v(k) a_k  +\overline{v(k)}a^{\dagger}_k) dk$. This can be easily diagonalized in the $L^2$ part of the domain, so that one has the fiber decomposition $\mathbb{H}\simeq \int_{\oplus} \mathbb{H}_P dP$ with a family of Hamiltonians acting only on  Fock space
\begin{equation}\label{HP}
\mathbb{H}_P:=\frac{1}{2m}(P-P_f)^2+\mathbb{F}+\sqrt{\alpha}\mathbb{V}
\end{equation} describing the system moving with momentum $P\in \mathbb{R}^d$. 
In this work, we are concerned with the ground state energies at fixed momentum, 
\begin{equation}
E(P):=\inf \mathrm{spec} ~\mathbb{H}_P
\end{equation} and the absolute ground state energy \begin{equation}
E_0=\inf\mathrm{spec}~\mathbb{H}=\inf_P E(P).
\end{equation} 
The following terminology concerning the dispersion relation will be useful below.  
\begin{enumerate}
  \item We say that $\epsilon$ is \emph{massive} if  $\Delta:=\inf_k \epsilon(k)>0$.  
  \item $\epsilon$ is \emph{subadditive} if $\epsilon(k_1+k_2)\leq \epsilon(k_1)+\epsilon(k_2)$ for all $k_1,k_2\in \mathbb{R}^d$.
  \item Moreover, we say that $\epsilon$ is of \emph{superfluid type} if 
  \begin{equation}
   \inf_{k\in \mathbb{R}^d}\frac{\epsilon(k)}{|k|}=:c >0.
  \end{equation} The number $c$ is called the \emph{critical velocity}.
\end{enumerate}
Prime examples of the above are optical phonons (with constant dispersion relation) for a massive and subadditive field and acoustic phonons (where $\epsilon(k)$ is linear in $|k|$) for a field of superfluid type. Physically, the first case is encountered in the original Fr\"ohlich polaron model, while a superfluid-type field is found in the Bose polaron. If the dispersion relation is massive and subadditive,  $E(P)$ is an isolated, simple eigenvalue for $P^2< 2m \Delta$. Moreover, $\inf_PE(P)=E(0)$, and $E(P)$ is an analytic function close to $P=0$ \cite{Gerlach, Moeller}. The {\em effective mass} is then defined as
\begin{equation}\label{def:Meff}
  M_{\text{eff}}:=\frac{1}{2}\lim_{P\rightarrow 0 } \left(\frac{E(P)-E(0)}{P^2}\right)^{-1}.
\end{equation}
In other words, $E(P) \approx E(0) + \frac{P^2}{2 M_{\rm eff}}$ for small $P$, and the system is envisioned as behaving, for sufficiently small momenta,  like  
a free particle of mass $M_{\mathrm{eff}}$ called the polaron, whence the entire model bears its name. 
We also introduce the function 
\begin{equation}\label{def:MP}
  M(P)=\frac{1}{2} \left(\frac{E(P)-E(0)}{P^2}\right)^{-1}
\end{equation} which we call \emph{the effective mass quotient}. It is well-defined for all $P$ s.t. $E(0)\neq E(P)$, and can be viewed as a \emph{global} measure of the curvature of $E(P)$, in contrast to $M_{\text{eff}}=\lim_{P\rightarrow 0}M(P)$ which quantifies this curvature \emph{locally} at $P=0$. The validity of the polaron picture can be also expressed as a statement that $M(P)$ is asymptotically a constant function for sufficiently small momenta. We find this picture useful below, where we shall consider both the case of $P$ vanishingly small as well as admitting values from a specified range.
\subsection{Motivation and statements of the results} We shall provide bounds on the above quantities in the limit of large $\alpha$. These bounds agree with the semi-classical approximation, which we now briefly recall. To do so, let us first observe from \eqref{HP} that the presence of the  particle induces non-trivial correlations between the modes of the field; if these are ignored, the problem is easily solvable. Indeed, in the case $m=\infty$, the spectrum of $\mathbb{H}_P$  is equal to that of the operator $\mathbb{F}-\alpha\|h\|^2$, with ground state energy $-\alpha \|h\|^2$. This corresponds to a free bosonic field fluctuating on top of a classical deformation profile induced by a point impurity. 
The ground state is simply the coherent state $|\phi\rangle $ with 
\begin{equation}
a_k|\phi\rangle = -\sqrt{\alpha}\frac{\overline{v(k)}}{\epsilon(k)}|\phi\rangle \qquad \forall k \in \mathbb{R}^d.
\end{equation} 
The evaluation of $\mathbb{H}$ on pure tensor products of the form $\psi\otimes \phi,$ where $\phi$ is a coherent state, amounts to replacing the creation and annihilation operators by complex numbers, which is equivalent to treating the boson field in a classical way. In fact, as is well known (we reproduce the argument in the proof of the upper bound in Theorem \ref{thm1} below), this coherent state ansatz is optimal over all product trial states. In other words, for polaron models, \emph{the adiabatic limit} (corresponding to a product trial state) \emph{and the strong-coupling limit coincide}. The adiabatic limit  can certainly be expected to be asymptotically correct if the mass of the particle is large. At the same time, it is well known that the adiabatic limit is asymptotically correct as $\alpha\rightarrow\infty$ in the Fr\"ohlich case, indirectly through  energy estimates \cite{LT} and also as far as the dynamics is concerned \cite{Dav1, Dav2}. If we were, therefore, to assume that the same conclusion is valid in more generality, our regular case included, we expect that
\begin{equation}\label{e1}
  \lim_{\alpha\rightarrow\infty}\frac{E_0}{\alpha}=-\|h\|^2.
\end{equation} 
Moreover, one can readily postulate how the next order correction should look like: since the leading order corresponds to the picture of a classical point particle situated at the bottom of a potential well created by the phonons, the next order correction should stem from the zero-point oscillations in this well. If we replace the  annihilation operators in $\mathbb{H}$ by the numbers $-\sqrt{\alpha}\frac{\overline{v(k)}}{\epsilon(k)}$ (and  $a_k^{\dagger}$ by its complex conjugate) and expand the $e^{ik\cdot x}$ factors to second order, we arrive at the one-particle Schr\"odinger operator \begin{equation}
  -\frac{\Delta_x}{2m}+\frac{m\omega^2}{2}x^2-\alpha \|h\|^2, 
\end{equation} where 
\begin{equation}\label{omega}
  \omega=\sqrt{\frac{2\alpha}{dm}}\|\nabla h\|,
\end{equation} with well-known ground state energy $\frac{d\omega}{2}-\alpha\|h\|^2$. We hence expect the subleading term to be $d\omega/2$, and thus of order $\alpha^{1/2}$. That these considerations are correct is the content of our first theorem. \subsubsection{Ground state energy asymptotics}
\begin{thm}\label{thm1} 
 Let $\mathbb{H}$ be regular. Then we have 
 \begin{equation}\label{bd}
 \begin{split}
 &  -\alpha \|h\|^2 + \sqrt{\frac{d \alpha}{2m}}\|\nabla h\|  ~\geq ~\inf \mathrm{spec}~\mathbb{H} ~\geq \\ & \geq  ~ -\alpha \|h\|^2 + \sqrt{\frac{d \alpha}{2m}}\|\nabla h\|-\frac{d}{2} \frac{\|\nabla \eta\|^2}{\|\nabla h\|^2} -\frac{d}{8m} \frac{\|\Delta h\|^2}{\|\nabla h\|^2} .
   \end{split}
 \end{equation} 
 In particular, for $E_0=\inf \mathrm{spec}~\mathbb{H}$, \eqref{e1} holds, and 
\begin{equation}
   \lim_{\alpha\rightarrow\infty} \alpha^{-1/2} \left(E_0+\alpha \|h\|^2\right)=\sqrt{\frac{d}{2m}}\|\nabla h\|.
 \end{equation}
\end{thm}
\begin{rem}
 As recalled in detail in the proof, the semiclassical limit naturally gives rise to the {\em Pekar functional}
 \begin{equation}\label{Pek}
  \mathcal{E}^{\rm{Pek}}_{\alpha}(\psi)=\frac{1}{2m}\int |\nabla\psi(x)|^2dx-\alpha\iint |\psi(x)|^2g(x-y)|\psi(y)|^2 dxdy,
\end{equation}where 
\begin{equation}\label{g}
g(x)=\int \frac{|v(k)|^2}{\epsilon(k)}e^{ik \cdot x}dk=(\bar{h}\ast h)(x),  
\end{equation} 
and 
\begin{equation}\label{pekar_energy}
E_0\leq E^{\mathrm{Pek}}:=\inf_{\psi: \|\psi\|=1}\mathcal{E}^{\rm{Pek}}_{\alpha}(\psi).
\end{equation} 
In the Fr\"ohlich case, where $g(x)=\frac{1}{|x|}$, one has $\inf \mathcal{E}_{\alpha}^{\rm{Pek}}(\psi)=\alpha^{2}\inf \mathcal{E}^{\rm{Pek}}_1(\psi)\equiv \alpha^2 e^{\mathrm{Pek}}$, in particular the contribution of the kinetic energy $\int |\nabla\psi|^2$ to $\mathcal{E}^{\rm{Pek}}_\alpha$ is \emph{not} negligible in this limit. Existing results \cite{LT} show that in this particular case 
\begin{equation}\label{kminus1}
  e^{\rm{Pek}}\geq \alpha^{-2} E_0 \geq e^{\mathrm{Pek}}-C\alpha^{-1/5}
\end{equation} for some $C>0$ and $\alpha$ large. If instead of $\mathbb{R}^d$ one considers a sufficiently regular subset $\Omega$ thereof (suitably rescaled to be of linear size $\alpha^{-1}$), with the corresponding modification  of $\eta$ involving  the  Laplacian on $\Omega$, then also the subleading correction to $E_0$, being of order $\alpha^0$, has been rigorously established \cite{FrankSe, DarioSe}. Adapting some of the methods in that proof to improve the control on the UV divergence of the model, the exponent $-1/5$ in the lower bound \eqref{kminus1} can be slightly improved to $-20/73$ \cite{mys19}.\\ \indent In our case, we lose the scaling properties of the original Fr\"ohlich model, and the semiclassical energy is a more general function of $\alpha$; our result captures the first two terms that emerge from the expansion of the kernel \eqref{g} of the Pekar functional  around its maximum. \end{rem}

\begin{rem} It can be argued that the $O(1)$ correction in \eqref{bd} is optimal as far as the order of magnitude is concerned. 
These $O(1)$ corrections can be attributed to two sources: the quantum fluctuations of the field and a purely classical effect of anharmonicity of the actual potential well that accompanies the particle's motion. 
We do not know how to obtain the sharp $O(1)$ correction to $E_0$, however.
\end{rem}

\begin{rem}\label{rem13}
While the lower bound in Theorem \ref{thm1} holds for all values of the parameters in the problem, it is optimal only in our case of interest, i.e., for large $\alpha$ with $m$  fixed. However, our analysis leads to a distinct result in the opposite regime with $m$ large at $\alpha$ fixed. In this case our technique yields a positive term $\frac{\alpha}{4m}\frac{\|\nabla h\|^4}{\|\nabla \eta\|^2}$ as a leading-order finite mass correction to the exact ground state energy at infinite mass $-\alpha\|h\|^2$ in the lower bound. Perturbation theory predicts here $-\alpha\|h\|^2+\frac{\alpha}{2m}\int \frac{k^2|v(k)|^2}{\epsilon(k)^2}dk+o(m^{-1})$. Clearly $\|\nabla h\|^4\leq \|\nabla \eta\|^2 \left(\int k^2 \frac{|v(k)|^2}{\epsilon(k)^2}dk\right)$, with equality in the case a constant dispersion relation, but even in this case the resulting correction is off by a factor of two with respect to the result from perturbation theory. 
\end{rem}

The proof of Theorem~\ref{thm1} will be given in Section~\ref{sec:21}. 
The upper bound relies on a straightforward expansion of the kernel of the Pekar functional. For the lower bound we  closely follow the approach of Lieb and Yamazaki \cite{LY} in their analysis of the Fr\"ohlich model. In the regular case considered here, the obtained bounds turn out to be sharp, however.

\subsubsection{Divergence of the effective mass}

Our further results concern the effective mass problem (for other rigorous work concerning this problem, we refer to \cite{DeckerPizzo, Spohn, LiebLoss, LiSe19} and references therein). First, we present a generalization of \cite{LiSe19} by showing that the effective mass diverges as $\alpha\rightarrow\infty$, in all spatial dimensions, and for all regular polaron models with massive fields. This is to be expected from the fact that the strong-coupling limit and the adiabatic limit coincide: while the particle's mass is fixed, the relevant dynamical degrees of freedom behave like a free particle with very large mass, which leads to the effective separation of timescales of the field and of the particle. 

\begin{thm}\label{thm2}
Let $\mathbb{H}$ satisfy the assumptions of Theorem \ref{thm1}, with a dispersion relation that is massive and subadditive. Then there exists a constant $C>0$ s.t. for all $\alpha\gg 1$ we have 
\begin{equation}
    M_{\mathrm{eff}}\geq C\alpha^{1/4}. 
  \end{equation}
\end{thm}
\begin{rem}
The assumption of subadditivity of $\epsilon$ is only used to ensure the existence of a ground state of $\mathbb{H}_0$, which is proved in \cite{Moeller}.
\end{rem}
\begin{rem}
We emphasize that the result holds regardless of the spatial dimension, assuming, of course, the required regularity of $\mathbb{H}$. On the other hand, in the physics literature the effective mass of a polaron has been investigated numerically also for various different polaron-type models \cite{PenaArdila, Pastukhov, Timour}, and it appears that for some of these models one can expect different behavior of the effective mass in different dimensions.
\end{rem}

Our proof takes the formula for the inverse of the effective mass from second-order perturbation theory as a starting point, and provides an upper bound on this quantity. The proof of this bound relies heavily on Theorem \ref{thm1}. It can be shown that the same conclusion can be reached with the method from \cite{LiSe19}. The regularity enables us to simplify the argument, and also to provide an explicit estimate on the rate of the divergence. However, based on the semiclassical analysis, we expect that the effective mass should actually, in the regular case, diverge much faster, namely linearly in $\alpha$. We perform this semiclassical analysis subsequently (see also \cite{Tempere}) before stating our last result, which addresses the effective mass quotient at non-zero $P$. 
 
 \subsubsection{Semiclassical analysis of the effective mass} 
 
As discussed above, the behavior of the system at strong coupling can be expected to be inferable from the semiclassical functional 
\begin{equation}\label{class}
  \frac{1}{2m}\int |\nabla \psi(x)|^2dx + 2\sqrt{\alpha}\mathfrak{Re} \int v(p)\varphi(p)\rho_\psi(p) dp +\int \epsilon(p)|\varphi(p)|^2 dp
\end{equation} 
with $\rho_\psi(p)=\int |\psi(x)|^2e^{ip\cdot x}dx$ and $\varphi:\mathbb{R}^d \to \mathbb{C}$ the classical field, which carries  momentum $\int p|\varphi(p)|^2 dp$. We wish to minimize \eqref{class} under the constraint that the total  momentum of the system be $P$. 
With $u$ a Lagrange multiplier (which can be interpreted as the velocity), the relevant functional to be minimized is thus 
\begin{equation}\label{hfr}
\begin{split}
  \mathcal{H}_P(\psi,\varphi,u)&=\int \epsilon(p)|\varphi(p)|^2dp+\frac{1}{2m}\int p^2 |\hat{\psi}(p)|^2dp + \\& \quad + 2\sqrt{\alpha}\mathfrak{Re} \int v(p)\varphi(p)\rho_\psi(p) dp +u\cdot \left(P-\int p\left(|\varphi(p)|^2+|\hat{\psi}(p)|^2\right) dp\right).
  \end{split}
\end{equation}
We expect continuity and accordingly $u\rightarrow 0$ as $P\rightarrow 0$; we also suppose that to leading order in $|u|$, the particle moves with velocity $u$ while maintaining its waveform, i.e., the $\psi$ minimizing \eqref{hfr} is approximately 
\begin{equation}
  \hat{\psi}_u(p)=\widehat{e^{im u\cdot x}\psi^{\mathrm{Pek}}_{\alpha}}(p) = \hat\psi_\alpha^{\rm Pek} (p - mu)
\end{equation} where $\psi^{\mathrm{Pek}}_{\alpha}$ minimizes \eqref{Pek}. Plugging this into \eqref{hfr}, we minimize with respect to the field, with the result that
  \begin{equation}
  \varphi_u(p)= -\sqrt{\alpha}\frac{\overline{\rho^{\mathrm{Pek}}_\alpha(p)v(p)}}{\epsilon(p)-u\cdot p}
\end{equation}  
with $\rho^{\rm Pek}_\alpha = \rho_{\psi^{\rm Pek}_\alpha}$. 
Accordingly, the Lagrange multiplier $u$ has to be chosen such that
 \begin{equation}\label{vp}
  P=mu+\alpha\int p \frac{|v(p)|^2|\rho_\alpha^{\mathrm{Pek}}(p)|^2}{(\epsilon(p)-u\cdot p)^2}dp.
\end{equation} 
Expanding this to leading order in $u$, we have
\begin{equation}\label{pv}
  P=\left(m+\frac{2\alpha}d \int p^2 \frac{|v(p)|^2|\rho_\alpha^{\mathrm{Pek}}(p)|^2}{\epsilon(p)^3}\right)u.
\end{equation} 
We further evaluate the energy $\mathcal{H}_P(\psi_u,\varphi_u,u)$ and expand it to second order in $P$, with the result that 
 \begin{equation}
\mathcal{H}_P(\psi_u,\varphi_u,u)\approx E^{\mathrm{Pek}}+\left(2m+\frac{4\alpha}{d}\int p^2 \frac{|v(p)|^2|\rho_\alpha^{\mathrm{Pek}}(p)|^2}{\epsilon(p)^3}dp\right)^{-1}P^2 
\end{equation}  with $E^{\mathrm{Pek}}$ defined in \eqref{pekar_energy}.
We are thus led  to the definition of the \emph{Pekar mass formula} \begin{equation}\label{pekmass}
  M^{\mathrm{Pek}}_{\alpha}:=\frac{2\alpha}{d}\iint |\psi^{\mathrm{Pek}}_{\alpha}(x)|^2R(x-y) |\psi^{\mathrm{Pek}}_{\alpha}(y)|^2dxdy=\frac{2}{d}\int p^2 \frac{|{\varphi}_{\alpha}^{\mathrm{Pek}}(p)|}{\epsilon(p)}^2~dp
\end{equation} where $R(x):=\int \frac{p^2|v(p)|^2}{\epsilon(p)^3}e^{ip\cdot x}dp$, $\psi^{\mathrm{Pek}}_{\alpha}$ minimizes \eqref{Pek} and $\varphi_{\alpha}^{\mathrm{Pek}}$ is the corresponding minimizing field. If we evaluate the above expression for the original Fr\"ohlich model, in which case  $\psi^{\mathrm{Pek}}_{\alpha}(x)=\alpha^{3/2}\psi^{\mathrm{Pek}}_1(\alpha x)$ and $R(x)=4\pi\delta(x-y)$, we obtain the celebrated \emph{Landau--Pekar mass formula} \cite{LanPek} \begin{equation}\label{LP}
M^{\mathrm{LP}}=\frac{8\pi}{3}\alpha^4 \|\psi^{\mathrm{Pek}}_1\|_4^4.
\end{equation} In the regular case we expect, given Theorem \ref{thm1}, that $|\psi^{\mathrm{Pek}}_{\alpha}|^2$ tends to a $\delta$-function as $\alpha\rightarrow\infty$, and accordingly that
\begin{equation}\label{def:MPek}
  \lim_{\alpha\rightarrow\infty}\alpha^{-1}M^{\mathrm{Pek}}_{\alpha}=\frac{2}{d}\int p^2\frac{|v(p)|^2}{\epsilon(p)^3}dp=:M^{\mathrm{Pek}}
\end{equation}  holds true. This leads us to the following \begin{conjecture}\label{conj}
Let $\mathbb{H}$ be regular. 
Then \begin{equation}\label{mass}
    \lim_{\alpha\rightarrow\infty}\alpha^{-1}M_{\mathrm{eff}}=M^{\mathrm{Pek}}=\frac{2}{d}\int \frac{p^2 |v(p)|^2}{\epsilon(p)^3}dp .
  \end{equation}
\end{conjecture}
This conjecture generalizes the one for the original Fr\"ohlich model, where one expects that $\lim_{\alpha\rightarrow\infty}\alpha^{-4}M_{\mathrm{eff}}=\frac{8\pi}{3}\|\psi^{\mathrm{Pek}}_{1}\|_4^4$, c.f. Eq. \eqref{LP}, as suggested by a calculation by Landau and Pekar \cite{LanPek}. A proof of this prediction remains an outstanding open problem. In the physics literature one also encounters discussions beyond the Fr\"ohlich case that lead to Conjecture \ref{conj} \cite{Tempere} (see also \cite[Eq. 12]{MiszaEnderalp}) and also to the linear dependence of the effective mass on $\alpha$ at strong coupling \cite{Grusdt}.

While we are unable to prove Conjecture \ref{conj}, we are able to prove a related result that can be regarded as a confirmation of the validity of the semiclassical approximation in the effective mass problem. Recall the definition of the effective mass quotient in \eqref{def:MP}.
Instead of the limit $\lim_{\alpha\rightarrow\infty}\alpha^{-1}\lim_{P\rightarrow 0}M(P)$, we consider the combined limit 
\begin{equation}\label{limit}
   \alpha\rightarrow\infty, \quad |P|\to \infty \quad \text{with} \quad |P|/\alpha \rightarrow 0, \quad |P|/\alpha^{1/2} \rightarrow \infty.
\end{equation} which we denote as \[\lim\limits_{\substack{\alpha\rightarrow\infty \\ \alpha^{1/2}\ll |P|\ll \alpha}}.\] Then we have 
\begin{thm}\label{thm3}
  Let $\mathbb{H}$ be regular, and assume that $\epsilon$ is massive and of superfluid type. Then 
  \begin{equation}\label{st1}
    \lim_{\substack{\alpha\rightarrow\infty \\ \alpha^{1/2}\ll |P|\ll \alpha}}\alpha^{-1}M(P)=M^{\mathrm{Pek}}
  \end{equation}
  with $M^{\rm Pek}$ defined in \eqref{def:MPek}. 
  In particular, we have for $\mathbb{H}$ satisfying the assumptions of Theorem \ref{thm1} and for all $P$ with $|P|\leq C\alpha$ for some $C>0$ independent of $P$ and $\alpha$, 
  \begin{equation}\label{mass_upper}
    E(P)\leq -\alpha\|h\|^2+\frac{d\omega}{2}+\frac{P^2}{2\alpha M^{\mathrm{Pek}}}+O(|P|  \alpha^{-1}) ,
  \end{equation}
  where $\omega$ is defined in \eqref{omega}. 
  If in addition the dispersion relation is assumed to be  of superfluid type, we have for all $P$ such that $|P|\leq C' \alpha$ with $C'>0$ small enough, 
  \begin{equation}\label{mass_lower}
    E(P)\geq -\alpha\|h\|^2+\frac{d\omega}{2}+\frac{P^2}{2\alpha M^{\mathrm{Pek}}}-\frac{d}{2}\frac{\|\nabla \eta\|^2}{\|\nabla h\|^2}-\frac{d}{8m}\frac{\|\Delta h\|^2}{\|\nabla h\|^2}-O(P^2\alpha^{-3/2}+|P|^3\alpha^{-2}) .
  \end{equation}
\end{thm} 
Note that \eqref{mass_upper} and \eqref{mass_lower} are non-zero momentum analogs of the bounds in Theorem \ref{thm1}. In combination, they readily imply \eqref{st1}. We conjecture that \eqref{st1} holds without the restriction that $|P|\gg \alpha^{1/2}$ (and hence, in particular, in the case when $P\to 0$ before $\alpha\to\infty$). 

Since the limit \eqref{limit} may at first sight appear artificial, let us briefly explain its origin. Theorem \ref{thm3} states that $E(P)-E(0)$ is, for large $\alpha$, and in a suitable window of momenta, asymptotically a parabolic curve with a coefficient determined by the semiclassical approximation. This can be regarded as a statement on the global curvature of $E(P)-E(0)$, in contrast to the \emph{local} curvature at $P=0$ described by the effective mass. 
The size of the window of momenta for which this asymptotic form holds depends on $\alpha$: the lower margin $|P|\gg \alpha^{1/2}$ ensures that we look at  $E(P)$ in the regime when the kinetic energy  of the center-of-mass motion is much larger than the $O(1)$ energetic error  determining the accuracy of our knowledge of the ground state energy, as expressed in Theorem \ref{thm1} and the bounds \eqref{mass_upper} and \eqref{mass_lower}. 
The upper margin  $|P|\ll \alpha$ is natural in view of the following discussion. Typically, one can expect that $E(P)$ has a parabolic shape for sufficiently small momenta, and this parabolic shape is in general lost when $E(P)$ approaches the bottom of the essential spectrum $E_{\mathrm{ess}}(P)$. 
There is a formula for the latter \cite{Moeller, Gerlach},
\begin{equation}
  E_{\mathrm{ess}}(P)=\inf_k \left(E(P-k)+\epsilon(k)\right),
\end{equation}
and in particular $E_{\rm ess}(P) \leq E(0) + \epsilon(P)$. Hence $E(P)\approx E(0) + \frac{P^2}{2 M_{\rm eff}}$ certainly ceases to be valid for $P^2 \gtrsim M_{\rm eff} \epsilon(P)$. Since $\epsilon(P) \geq c|P|$ by assumption, this is thus the case if $|P| \gtrsim M_{\rm eff} \sim \alpha$. 
 
 Theorem \ref{thm3} may be regarded as our principal novel contribution to the existing literature. Its proof utilizes, in particular,  a new trial state in order to obtain the upper bound \eqref{mass_upper}, which is essentially the extension of the very simple bound $E(0)\leq \inf_{\psi}\mathcal{E}^{\mathrm{Pek}}_\alpha(\psi)$ to non-zero momenta. The lower bound relies, on the other hand, on an extension of the techniques used in the lower bound of Theorem \ref{thm1}, and is thus also ultimately rooted in \cite{LY}. 
 
\bigskip

 In the remainder of the article we give the proofs of our results. The symbol $C$ denotes a positive constant, independent of $\alpha$ and $P$, whose exact value may change from one instance to the other. 

\section{Proofs}
\subsection{Proof of Theorem \ref{thm1}}\label{sec:21}

\subsubsection{Upper bound}
\begin{proof} 

For any normalized $\phi\in \mathcal{F}$, we have
\begin{equation}
  \langle \phi|a_k^{\dagger}a_k|\phi\rangle \geq |\langle\phi|a_k|\phi\rangle |^2 \qquad \forall k\in \mathbb{R}^d 
\end{equation} 
with equality if and only if $\phi$ is a coherent state, i.e., an eigenstate of all the $a_k$. Thus
 \begin{equation}\label{sem}
\inf_{\phi, \psi }\langle \psi \otimes \phi |\mathbb{H}| \psi\otimes \phi \rangle=\inf_{\varphi,\psi}\mathcal{H}(\psi,\varphi)
\end{equation} 
where $\mathcal{H}$ is the classical functional 
\begin{equation}
  \mathcal{H}(\psi,\varphi)=\frac{1}{2m}\int |\nabla \psi(x)|^2dx + 2\sqrt{\alpha}\mathfrak{Re} \int v(p)\varphi(p)\rho_\psi(p) dp +\int \epsilon(p)|\varphi(p)|^2 dp
\end{equation} 
with $\rho_\psi(p)=\int |\psi(x)|^2 e^{ip\cdot x}dx$.  Minimizing with respect to the field $\varphi$ and passing to position space, we obtain the \emph{Pekar functional}
\begin{equation}
  \mathcal{E}^{\rm{Pek}}_{\alpha}(\psi)=\frac{1}{2m}\int |\nabla\psi(x)|^2dx-\alpha\iint |\psi(x)|^2g(x-y)|\psi(y)|^2 dxdy,
\end{equation} with $g(x)=\int \frac{|v(k)|^2}{\epsilon(k)}e^{ik\cdot x}dk.$  
Since $\mathbb{H}$ is isotropic, $g(x)=\int \frac{|v(k)|^2}{\epsilon(k)}\cos(k\cdot x) dk$. By the elementary inequality $\cos x \geq 1-\frac{x^2}{2}$ we have \begin{equation}
  \mathcal{E}^{\rm{Pek}}_{\alpha}(\psi)\leq -\alpha\|h\|^2+\mathcal{L}_{\alpha}(\psi)
\end{equation}
with the functional 
\begin{align}
  \mathcal{L}_{\alpha}(\psi) & =\frac{1}{2m}\int |\nabla\psi(x)|^2 dx +\frac{\alpha \|\nabla h\|^2}{2d}\iint |\psi(x)|^2 (x-y)^2 |\psi(y)|^2 dx dy \nonumber \\ & = \frac{1}{2m}\int |\nabla\psi(x)|^2 dx +\frac{\alpha \|\nabla h\|^2}{d}\left(\int |\psi(x)|^2 x^2-\left(\int x|\psi(x)|^2 dx\right)^2\right).
\end{align}
It follows from the Heisenberg uncertainty principle  that the infimum of $\mathcal{L}_{\alpha}$ equals $d\omega/2$, with $\omega$ in \eqref{omega}. This leads to the claimed upper bound. \end{proof}
\subsubsection{Lower bound}
\begin{proof}
Our starting point is the inequality, valid for all $\lambda \in \mathbb{R}$ and all $R\in \mathbb{R}^d$, 
\begin{equation}\label{cs}
\sqrt{\alpha}  \lambda[P_f+R,[\mathbb{V},P_f+R]]\leq \frac{1}{2m}(P_f+R)^2-2m\lambda^2 \alpha [P_f+R,\mathbb{V}]^2
\end{equation} 
which can be easily proved using the Cauchy--Schwarz inequality (the minus sign on the right-hand side stems from the fact that $[P_f,\mathbb{V}]$ is anti-hermitian). We have the identity 
\begin{equation}
\lambda[P_f+R,[\mathbb{V},P_f+R]]= \lambda[P_f,[\mathbb{V},P_f] =  -\lambda\int k^2 \left(\overline{v(k)}a^{\dagger}_k+v(k)a_k\right)dk.
\end{equation} 
We  conclude that for all $P\in \mathbb{R}^d$, the operators $\mathbb{H}_P$ in \eqref{HP} are bounded below, uniformly in $P$, by \begin{equation}
  \mathbb{H}_P\geq \mathbb{H}'_{\lambda}:=\mathbb{F}+\sqrt{\alpha}\mathbb{W}_{\lambda}+2m\lambda^2\alpha [P_f,\mathbb{V}]^2
\end{equation}
with \begin{equation}
  \mathbb{W}_{\lambda}=\int (1-\lambda k^2)\left(v(k)a_k+\overline{v(k)}a^{\dagger}_k\right) dk.
\end{equation}
Given that $\mathbb{H}$ is unitarily equivalent to $\int_{\oplus} \mathbb{H}_P dP$, this clearly implies that
 \begin{equation}
  \inf \mathrm{spec}~ \mathbb{H} \geq \sup_{\lambda} ~\inf \mathrm{spec}~ \mathbb{H}'_{\lambda}.
\end{equation}
Let 
\begin{equation}\label{w}
w_\lambda(k) := (1-\lambda k^2) v(k)  
\end{equation} 
and observe that our assumptions on $v$ and $\epsilon$, i.e.,  $h\in W^{2,2}(\mathbb{R}^d)$ and $\epsilon$ massive,  imply that $w_\lambda\epsilon^{-1}\in L^2(\mathbb{R}^d)$. We may thus apply the unitary shift operator $U$ with the property that $U a_k U^{\dagger}=a_k-\sqrt{\alpha}\frac{\overline{w_\lambda(k)}}{\epsilon(k)}$ for all $k\in \mathbb{R}^d$ to $\mathbb{H}_{\lambda}$. We obtain 
\begin{equation}\label{wicht}
  U \left(\mathbb{F}+\sqrt{\alpha}\mathbb{W}_{\lambda}\right) U^{\dagger}=\mathbb{F}-\alpha\|h\|^2+2\lambda\alpha\|\nabla h\|^2-\lambda^2\alpha \|\Delta h\|^2 .
\end{equation}
Furthermore 
\begin{align} \nonumber 
U[P_f,\mathbb{V}] U ^{\dagger} & =\int k v(k)\left(a(k)-\alpha^{1/2}\overline{w_{\lambda}(k)}\epsilon(k)^{-1}\right)dk-\int k\left( \overline{v(k)}a^{\dagger}_k-\alpha^{1/2}w_{\lambda}(k)\epsilon(k)^{-1}\right)dk \\ &=[P_f,\mathbb{V}] 
  \label{pfv}
\end{align}
since $\int k v(k) \overline{w_{\lambda}(k)}\epsilon(k)^{-1} dk \in \mathbb{R}$.  We are thus left with providing a lower bound to the operator 
\begin{equation}\label{FPFV}
  \mathbb{F}+2m\lambda^2 \alpha [P_f,\mathbb{V}]^2. 
\end{equation}
Since $\eta\in W^{1,2}(\mathbb{R}^d)$ by assumption, we can introduce the bosonic operators, for $i=1,\dots,d$, 
\begin{equation}
  b_i=\frac{\sqrt{d}}{\|\nabla \eta\|}\int k_i v(k) a_k dk
\end{equation} 
with $[b_i,b_j^{\dagger}]=\delta_{ij}$. Then \begin{equation}
  [P_f,\mathbb{V}]^2=\frac{\|\nabla \eta\|^2}{d}\sum_{i=1}^d(b_i-b_i^{\dagger})^2.
\end{equation}
Let 
\begin{equation}\label{E}
  E= \frac{\|\nabla \eta\|^2}{\|\nabla h\|^2} .
\end{equation} 
We claim that 
 \begin{equation}
  \mathbb{F}\geq E\sum_{i=1}^db_i^{\dagger}b_i.
\end{equation} 
To prove this, it is enough to show that for every one-phonon vector $\Phi \in L^2(\mathbb{R}^d)$ 
\begin{equation}\label{want}
  \int \epsilon(k)|\Phi(k)|^2 dk\geq \frac{dE}{\|\nabla \eta\|^2}\sum_{i=1}^d\left| \int k_i \overline{v(k)}\Phi(k)dk \right|^2.
\end{equation} 
For any $\psi\in L^2(\mathbb{R}^d)$ and $d$ orthonormal functions $\phi_i$ we have by  Bessel's inequality \begin{equation}
  \int |\psi(k)|^2dk \geq \sum_{i=1}^d \left|\int \overline{\phi_i(k)}\psi(k)dk \right|^2.
\end{equation} Using this for $\psi(k)=\sqrt{\epsilon(k)}\Phi(k)$ and $\phi_i(k)=(d^{-1/2}\|\nabla h\|)^{-1}k_i \frac{v(k)}{\sqrt{\epsilon(k)}}$ yields \eqref{want}. Moreover, since 
 \begin{equation}
  (b_i-b_i^{\dagger})^2=(b_i+b_i^{\dagger})^2-4b^{\dagger}_ib_i-2\geq -4b^{\dagger}_ib_i-2
\end{equation}
we conclude that \eqref{FPFV} is bounded below by $-4m\lambda^2\alpha\|\nabla \eta\|^2$ provided that \begin{equation}
  |\lambda|\leq \lambda_0=\sqrt{\frac{d}{8m\alpha \|\nabla h\|^2}}=\frac{1}{2m\omega}.
\end{equation} 
In particular, 
\begin{equation}
 \inf \mathrm{spec}~ \mathbb{H} \geq  -\alpha\|h\|^2+   \alpha \sup_{|\lambda|\leq \lambda_0} \left( 2\lambda\|\nabla h\|^2-\lambda^2 \|\Delta h\|^2  -4m\lambda^2\|\nabla \eta\|^2 \right) .
 \end{equation} 

The choice $\lambda= \lambda_0$ yields the lower  bound  in \eqref{bd}. This choice for $\lambda$ is in fact optimal as long as 
\begin{equation}
 \alpha\geq  \alpha_m:=\frac{d}{8m}\left(\frac{4m\|\nabla \eta\|^2+\|\Delta h\|^2}{\|\nabla h\|^3}\right)^2.
\end{equation}
For $\alpha<\alpha_m$, the optimal choice of $\lambda$ is rather $\lambda = \|\nabla h\|^2 ( 4m\|\nabla \eta\|^2+\|\Delta h\|^2)^{-1}$, which yields the improved lower bound mentioned in Remark \ref{rem13} in this case.
\end{proof}

\subsection{Proof of Theorem \ref{thm2}}

\begin{proof} 
Under the stated assumptions on $v$ and $\epsilon$, there exists an isolated eigenvalue at the bottom of the spectrum of $\mathbb{H}_0$, and a unique corresponding ground state $\phi_0$ \cite{Moeller}. 
Using second-order perturbation theory and rotation invariance, one arrives at the formula \begin{equation}
  \frac{1}{2M_{\text{eff}}}=\frac{1}{2m}-\frac{1}{dm^2}\langle \phi_0 |P_f\frac{1}{\mathbb{H}_0-E_0}P_f|\phi_0\rangle
\end{equation}
for the effective mass $M_{\rm eff}$ defined in \eqref{def:Meff}.
 Note that $\langle \phi_0 |P_f| \phi_0\rangle =0$, hence $QP_f|\phi_0\rangle=P_f|\phi_0\rangle$, where $Q$ is the projection onto the orthogonal complement of the ground state of $\mathbb{H}_0$, and $\mathbb{H}_0 - E_0$ is strictly positive and invertible on the range of $Q$. Therefore, by the Cauchy-Schwarz inequality, 
\begin{equation}\label{massbd}
  \frac{m}{M_{\text{eff}}}\leq 1 - \frac{2}{dm} \frac{ \langle \phi_0 | P_f^2|\phi_0\rangle^2 } { \langle \phi_0 | P_f (\mathbb{H}_0-E_0) P_f|\phi_0\rangle} . 
\end{equation} 
 We exploit the fact that $\phi_0$ is the ground state of $\mathbb{H}_0$ and arrive at the identity 
\begin{equation}
  \langle \phi_0 | P_f (\mathbb{H}_0-E_0) P_f|\phi_0\rangle=\frac{1}{2}\langle\phi_0| [P_f,[\mathbb{H}_0,P_f]]|\phi_0\rangle. 
\end{equation} 
A simple computation shows that the double commutator  equals 
\begin{equation}\label{ww}
  \frac{1}{2}[P_f,[\mathbb{H}_0,P_f]]=\mathbb{W}:=-\frac{\sqrt{\alpha}}{2}\int k^2 \left(v(k)a_k+\overline{v(k)}a_k^{\dagger}\right)dk.
\end{equation} 
Define for $\lambda>-\frac{1}{2m}$ and  $\mu\in \mathbb{R}$  \begin{equation}
  \mathbb{H}(\lambda,\mu):=\mathbb{H}_0+\lambda P_f^2+\mu \mathbb{W}, \quad E_0(\lambda,\mu)=\inf \mathrm{spec}~\mathbb{H}(\lambda,\mu).
\end{equation} By the variational principle, \begin{equation}\label{var}
E_0(\lambda,\mu)\leq E_0+\lambda \langle \phi_0 | P_f^2|\phi_0\rangle+\mu \langle \phi_0 |\mathbb{W}|\phi_0\rangle. \end{equation} We have the lower bound \begin{equation}
  E_0(0,\mu)\geq -\alpha\int \left(1-\frac{\mu k^2}{2}\right)^2\frac{|v(k)|^2}{\epsilon(k)}dk
\end{equation} 
which is finite because of our assumption $h \in W^{2,2}(\mathbb{R}^d)$. Combining the last two inequalities with the upper bound on $E_0$ from Theorem \ref{thm1}, we  conclude that for all negative $\mu$ 
\begin{equation}
  \langle \phi_0 |\mathbb{W}|\phi_0\rangle \leq \frac{dm\omega^2}{2} - \alpha\frac{\mu}{4}\int k^4 \frac{|v(k)|^2}{\epsilon(k)}dk - \mu^{-1}\frac{d\omega}{2}.
\end{equation} 
Optimizing over $\mu<0$ yields the bound
\begin{equation}\label{w_t2}
 \langle \phi_0 |\mathbb{W}|\phi_0\rangle \leq   \frac{dm\omega^{2}}{2} \left( 1+C\alpha^{-1/4} \right) 
\end{equation}
for $C= \|\Delta h\| m^{-1/4} (d/2)^{1/4} \|\nabla h\|^{-3/2}$. 

In a similar fashion we conclude  from \eqref{var} that for all $\lambda>-\frac{1}{2m}$ \begin{equation}\label{p1}
\lambda  \langle \phi_0 | P_f^2|\phi_0\rangle \geq E_0(\lambda,0)-E_0 \geq E_0(\lambda,0)+\alpha\int \frac{|v(k)|^2}{\epsilon(k)}dk-\frac{d\omega}{2}
\end{equation} where we again used the upper bound from Theorem \ref{thm1}. Now, a lower bound on $E_0(\lambda,0)$ is provided by Theorem \ref{thm1} for a particle with mass $\frac{m}{(1+2\lambda m)}$ :
\begin{equation}
  E_0(\lambda,0)\geq -\alpha \int \frac{|v(k)|^2}{\epsilon(k)}dk +\frac{d\omega}{2}\sqrt{1+2\lambda m}-\frac{d}{2} \frac{\|\nabla \eta\|^2}{\|\nabla h\|^2} -\frac{d(1+2\lambda m)}{8m} \frac{\|\Delta h\|^2}{\|\nabla h\|^2} .
\end{equation}
In particular, for any $\lambda>0$,
\begin{equation}\label{p4}
  \langle \phi_0 | P_f^2|\phi_0\rangle \geq \frac{dm\omega}{2}+\left(\frac{d\omega}{2}\left(\frac{\sqrt{1+2m\lambda}-1}{\lambda}-m\right)\right)-C\frac{1+\lambda}{\lambda}
\end{equation} 
for suitable $C>0$. For small $\lambda$, the second term on the right-hand side behaves like $\omega\lambda$, hence the optimal choice of $\lambda$ is of the order $\omega^{-1/2} \sim \alpha^{-1/4}$, and we arrive at the bound
\begin{equation}\label{p2}
  \langle \phi_0 | P_f^2|\phi_0\rangle \geq \frac{dm\omega}{2} \left( 1 - C \alpha^{-1/4} \right)  \qquad \text{for $\alpha\gg 1$}.
\end{equation}
Combining \eqref{massbd}, \eqref{w_t2} and \eqref{p2}, we arrive at the claimed lower bound on the effective mass.  \end{proof}

\subsection{Proof of Theorem \ref{thm3}}
\subsubsection{Lower bound}

\begin{proof}
Recall our assumption $\epsilon(k) \geq c |k|$ for some $c>0$. Pick $u\in \mathbb{R}^d$ with $|u|<c$, and write
\begin{equation}\label{inr}
  \mathbb{H}_P= P \cdot u - \frac m2 u^2 + \frac{(P_f-P + mu )^2}{2m}+\tilde{\mathbb{F}}_u+\sqrt{\alpha}\mathbb{V}
\end{equation} 
where 
\begin{equation}\label{FR}
 \tilde{\mathbb{F}}_u=\int \left(\epsilon(k) - u\cdot k \right) a^{\dagger}_k a_k dk. 
\end{equation}
We proceed as in the proof of Theorem~\ref{thm1} and use \eqref{cs}, this time for $R=P-mu$. This gives the lower bound
\begin{align}\nonumber 
  E(P)& \geq  P \cdot u - \frac m2 u^2 -\alpha \|h_u\|^2 \\ & \quad + \sup_{\lambda \in \mathbb{R}} \left[ \inf\mathrm{spec} \lbrace\tilde{\mathbb{F}}_u+2m\alpha \lambda^2 [P_f,\mathbb{V}]^2\rbrace+ 2\lambda\alpha\|\nabla h_u\|^2-\lambda^2\alpha \|\Delta h_u\|^2 \right] \label{Ep}
\end{align} 
with
\begin{equation}
h_u(x)=\frac{1}{(2\pi)^{d/2}}\int \frac{v(k)}{\sqrt{\epsilon(k)-u\cdot k}}e^{ik\cdot x}dk .
\end{equation} 
Without loss of generality, we can assume that $P=|P|e_1$, where $e_1$ is the unit vector pointing in the first coordinate direction, and we shall also pick $u$ to point along $e_1$. The functions 
\begin{equation}
  \phi_i^u(k)=\frac{k_i v(k)}{\sqrt{\epsilon(k)-u\cdot k}}=\frac{k_i v(k)}{\sqrt{\epsilon(k)- |u| k_1}}
\end{equation} 
are then orthogonal, and 
\begin{align}\nonumber
 \|\phi_i^u\|^2 & = \frac{1}{d}\int \frac{k^2|v(k)|^2}{\epsilon(k)}dk+ u^2 \int \frac{k_i^2k_1^2|v(k)|^2}{\epsilon(k)^2(\epsilon(k)- |u| k_1)}dk \\ & \leq\frac{1}{d}\int \frac{k^2|v(k)|^2}{\epsilon(k)}dk \left(1+ \frac{u^2}{c(c- |u|)} \right) 
 \end{align} 
 where we used $\epsilon(k)\geq c|k|\geq c|k_1|$  in the last step. 
 Bessel's inequality 
 \begin{equation}
 \int |\psi(k)|^2dk\geq \sum_{i=1}^d \left|\frac{1}{\|\phi^u_i\|}\int \overline{\phi^u_i(k)}\psi(k)dk \right|^2
\end{equation}
thus yields 
\begin{equation}
  \tilde{\mathbb{F}}_u\geq \left(1+ \frac{u^2}{c(c- |u|)} \right)^{-1} E\sum_{i=1}^d b_i^{\dagger}b_i
\end{equation} 
with $E$ defined  in \eqref{E}. 
Arguing as in the proof of the lower bound in Theorem \ref{thm1}, we conclude that 
\begin{equation}
\inf\mathrm{spec} \lbrace\tilde{\mathbb{F}}_u+2m\alpha \lambda^2 [P_f,\mathbb{V}]^2\rbrace \geq -4m\lambda^2\alpha\|\nabla \eta\|^2
\end{equation}
as long as 
\begin{equation}\label{lam}
|\lambda| \leq \left(1+ \frac{u^2}{c(c- |u|)} \right)^{-1/2}\frac{1}{2m\omega}  .
\end{equation}
We choose the maximally allowed value of $\lambda$ (i.e., equality in \eqref{lam}), and arrive at the lower bound 
\begin{align}\nonumber 
  E(P)& \geq  P \cdot u - \frac m2 u^2 -\alpha \|h_u\|^2 \\ & \quad  +  \left(1+ \frac{u^2}{c(c- |u|)} \right)^{-1/2} \sqrt{ \frac{\alpha d}{2m}} \|\nabla h_u\|^2  -  \frac{d}{2} \frac{ \|\nabla \eta\|^2}{\|\nabla h\|^2}  -  \frac{d}{8m} \frac{ \|\Delta h_u\|^2}{\|\nabla h\|^2} . \label{Ep2}
\end{align} 
We are left with estimating the norms appearing in \eqref{Ep2}. We have 
 \begin{equation}
 \|h_u\|^2=\int\frac{|v(k)|^2}{\epsilon(k)}\left( 1 + \frac{ (u\cdot k)^2 }{ \epsilon(k) \left( \epsilon(k) - u\cdot k\right)} \right) dk \leq \| h\|^2 + u^2 \frac{M^{\rm Pek}}{2 ( 1- |u|/c)}
\end{equation}
where we used the definition of $M^{\rm Pek}$ in \eqref{def:MPek} and $\epsilon(k)\geq c|k|$. Similarly,
\begin{equation}
 \| \Delta h_u\|^2=\int\frac{|k|^4 |v(k)|^2}{\epsilon(k)}\left( 1 + \frac{ (u\cdot k)^2 }{ \epsilon(k) \left( \epsilon(k) - u\cdot k\right)} \right) dk  \leq \| \Delta h\|^2 \left(1+ \frac{u^2}{c(c- |u|)} \right) . 
\end{equation}
 For the remaining term, we simply bound 
\begin{equation}
 \|\nabla h_u\|^2=\int\frac{k^2 |v(k)|^2}{\epsilon(k)}\left( 1 + \frac{ (u\cdot k)^2 }{ \epsilon(k) \left( \epsilon(k) - u\cdot k\right)} \right) dk \geq \| \nabla h\|^2  . 
\end{equation}
We are still free to choose $u$ (subject to the constraint $|u|<c$) and the leading terms to optimize are simply $P\cdot u - \alpha u^2 M^{\rm Pek}/2$. We therefore choose
\begin{equation}
u = \frac P { \alpha M^{\rm Pek}}
\end{equation}
which yields the desired bound \eqref{mass_lower}. 
\end{proof}

\begin{rem}
By choosing $u$ simply $O(1)$, the bound \eqref{Ep2} implies that  
under the same assumptions on $v$ and $\epsilon$, there exist $\gamma>0$ and $F_{\alpha}\in \mathbb{R}$ such that for all $P\in \mathbb{R}^d$, 
\begin{equation}
  E(P)\geq \gamma |P|+F_{\alpha}.
\end{equation} 
From this and from the analyticity of $E(P)$ in a neighborhood of its global minimum at $P=0$ one can deduce that there exists a $P^*> 0$ such that 
\begin{equation}\label{co_hull}
  E(P)=E^*(P), \quad \forall P: 0\leq |P|\leq P^*,
\end{equation} 
where $E^*$ denotes the convex envelope of $E$, i.e., the largest convex function not exceeding $E$. One can verify that \begin{equation}
  E^*(P)=\sup_{s\in\mathbb{R}^d}\left(sP+\inf_{Q\in \mathbb{R}^d}\left(E(Q)-sQ\right)\right);
\end{equation}
 using the Lee-Low-Pines transformation, this can be cast into the form
\begin{equation}\label{estar}
  E^*(P)=\sup_{s \in \mathbb{R}^d} \left(s \cdot P+\inf \mathrm{spec}~ \left(\mathbb{H}-s \cdot P_{\rm tot}\right)\right)
\end{equation} 
where $P_{\rm tot}=-{i}\nabla_x+P_f$. 
In particular, for any $ \varphi \in L^2(\mathbb{R}^d)$ and any $L^2$-normalized $\psi\in H^1(\mathbb{R}^d)$  we have $E^*(P)\leq \sup_{s} \mathcal{H}_P(\psi,\varphi,s)$, where $\mathcal{H}_P(\psi,\varphi,s)$ is defined in \eqref{hfr}. Choosing 
\begin{equation}\label{funcs}
\psi=\psi^{\mathrm{Pek}}_{\alpha}, \quad \varphi(p)=\varphi_{\alpha}^{\mathrm{Pek}}(p) \left(1+\frac{p\cdot P}{\epsilon(p)M^{\mathrm{Pek}}_{\alpha}}\right)
\end{equation} one easily arrives at 
\begin{equation}\label{from-cov_hull0}
 E^*(P)\leq E^{\mathrm{Pek}}+\frac{P^2}{2M^{\mathrm{Pek}}_{\alpha}} \qquad  \forall P \in \mathbb{R}^d ,
\end{equation} 
and hence 
\begin{equation}\label{from-cov_hull}
 E(P)\leq E^{\mathrm{Pek}}+\frac{P^2}{2M^{\mathrm{Pek}}_{\alpha}} \qquad  \text{$\forall P$ with $|P|\leq P^*$} ,
\end{equation} 
where $E^{\mathrm{Pek}}$ is the infimum of the Pekar functional \eqref{Pek}, and $M^{\mathrm{Pek}}_{\alpha}$ is given by the Pekar mass formula \eqref{pekmass}. If we knew that $|P^*|\sim\alpha$, we could already deduce the main statement of Theorem \ref{thm3}, Eq. \eqref{st1}. Without this knowledge, we need to find an upper bound directly on  $E(P)$  by using an appropriate trial state for $\mathbb{H}_P$, with is the topic of the next section. 
The resulting bound holds for all $|P|\lesssim \alpha$, and is hence sufficient for our purpose. Let us also emphasize that for the upper bound in \eqref{from-cov_hull} via the equality \eqref{co_hull} the superfluid property of $\epsilon$ is crucial. In fact, for a constant dispersion relation (and hence, in particular, in the Fr\"ohlich case) $E^*(P)\equiv E(0)$ and hence $P^* = 0$. On the other hand, the proof of the upper bound that we shall now give holds for all regular polaron Hamiltonians, without the restriction that $\epsilon$ be superfluid.

\end{rem}

\subsubsection{Upper bound}

\begin{proof}
\textit{\underline{The trial state}}:
Let $\psi\in L^2(\mathbb{R}^d)$ be real-valued, with Fourier transform in $H^1(\mathbb{R}^d)$, and let $\varphi\in L^2(\mathbb{R}^d)$. We denote by $|\varphi\rangle$  the coherent state corresponding to $\varphi$, satisfying $a_k |\varphi\rangle = \varphi(k) |\varphi\rangle$ for all $k\in \mathbb{R}^d$. Explicitly, $|\varphi\rangle=e^{a^{\dagger}(\varphi)-a(\varphi)}|\Omega\rangle$ with $|\Omega\rangle$ -- the vacuum on $\mathcal{F}$. We choose a trial state (on $\mathcal{F}$) of the form (comp. \cite{nagy,liebse_equi_mass}) \begin{equation}\label{state}
  |\phi_P\rangle=\psi(P-P_f)|\varphi\rangle.
\end{equation} 
This state corresponds to the $P$-momentum fiber of the product state $\psi\otimes| \varphi\rangle$. It appears that this particular form of a trial state for $\mathbb{H}_P$ was first considered, for the case $P=0$, by Nagy \cite{nagy}, who in this way obtained the bound $E(0)\leq E^{\mathrm{Pek}}$ directly on $\mathcal{F}$. This form is also behind the intuition of the trial state in \cite{liebse_equi_mass}, where its linearized version is considered. In these cases $\psi$ and $\varphi$ were chosen to be the momentum space minimizers of the Pekar functional. We shall rather choose functions related to the ones mentioned in the preceding Remark, i.e., \eqref{funcs}, in particular $\varphi$ will have an additional explicit $P$-dependence. Thanks to the regularity, we can sligthly simplify their form using the intuition from Theorem \ref{thm1}, which facilitates the computations. Note that \eqref{state} induces non-trivial correlations between different modes of the field, in contrast to the full product state. One of the main points of the analysis below is to show that these correlations lead to subleading corrections to the desired energy expression, which naturally appears for our choice of $\psi$ and $\varphi$. We proceed with the details and start by rewriting the expected value of the energy and the norm of our trial state in a suitable way.

\underline{\emph{Preliminary computations}}:
We have the identity \begin{equation}
   a_p\psi(P-P_f)=\psi(P-p-P_f)a_p
 \end{equation} whence we deduce the relations 
\begin{equation}
\psi(P-P_f)\mathbb{V}\psi(P-P_f)= \int \dd p ~v(p) \psi(P-P_f-p)\psi(P-P_f)a_p+ \text{h.c.}
\end{equation}
as well as 
\begin{equation}
  \begin{split}
  \psi(P-P_f)\mathbb{F}\psi(P-P_f)=\int \epsilon(p) a^{\dagger}_p \psi(P-P_f-p)^2 a_p \dd p.  
  \end{split}
\end{equation}
Consequently
\begin{align} \nonumber 
 \langle \phi_P|\mathbb{H}_P |\phi_P\rangle &=  \frac{1}{2m}\langle \varphi|(P-P_f)^2\psi(P-P_f)^2|\varphi \rangle +\int \dd p ~ \epsilon(p) |\varphi(p)|^2 \langle \varphi | \psi(P-P_f-p)^2|\varphi \rangle \\ &\quad + 2\sqrt{\alpha}\mathfrak{Re}\int \dd p~ v(p)\varphi(p) \langle \varphi | \psi(P-P_f-p)\psi(P-P_f)|\varphi\rangle .
\label{exp_values}
\end{align}
Define 
\begin{equation}
  \GG(R)=\langle \varphi|\psi(R-P_f)^2|\varphi\rangle.
\end{equation}
In particular, $\langle \phi_P|\phi_P\rangle=\GG(P)$. Using the properties of the Weyl operator $e^{a^{\dagger}(\varphi)-a(\varphi)}$, we compute
\begin{equation}
  \langle \varphi |e^{-ix \cdot P_f}|\varphi\rangle =\exp\left(\int |\varphi(p)|^2(e^{-ip\cdot x}-1)\dd p\right)
\end{equation} and obtain \begin{equation}\label{gg}
  \GG(R)=\frac{1}{(2\pi)^d}\int \dd x ~ \rho_{\psi}(x)e^{F(x)-F(0)+iR\cdot x}
\end{equation} where \begin{equation}
  \rho_{\psi}(x)=\int |\psi(k)|^2 e^{-ik\cdot x}~\dd k 
\end{equation}
and \begin{equation}
  F(x):=\rho_{\varphi}(x)=\int |\varphi(p)|^2 e^{-ip\cdot x} \dd p.
\end{equation}
In a similar fashion, we obtain \begin{equation}
  G^{(2)}_{\psi,\varphi}(R,S):= \langle \varphi | \psi(R-P_f)\psi(S-P_f)|\varphi\rangle=\frac{1}{(2\pi)^d}\int \rho^{(2)}_{\psi}(x;R-S)e^{iR\cdot x}e^{F(x)-F(0)}\dd x
\end{equation}
with 
\begin{equation}
  \rho^{2}_{\psi}(x;y)=\int \psi(k)\psi(k-y)e^{-ik\cdot x}\dd k.
\end{equation}
Finally,
\begin{equation}\label{def:T}
 \frac{1}{2m}\langle \varphi| (P-P_f)^2\psi(P-P_f)^2|\varphi\rangle = \frac{1}{(2\pi)^d}\int \tau_{\psi}(x) e^{iP\cdot x}e^{F(x)-F(0)}\dd x
\end{equation} 
where
\begin{equation}\label{tau_def}
  \tau_{\psi}(x)=\frac{1}{2m}\int k^2 \psi(k)^2e^{-ik\cdot x}\dd k. 
\end{equation} 
We shall now specify our choice of $\psi$ and $\varphi$. We choose 
\begin{equation}
  \psi(k)=e^{-\frac{k^2}{2m\omega}}
\end{equation} where $\omega$ is defined in \eqref{omega}. With this choice of $\psi$, we have 
\begin{equation}\label{rho2_ev}
      \rho^{2}_{\psi}(x;p)= \left(m \pi \omega \right)^{d/2}e^{-\frac{1}{4m\omega}p^2}e^{-\frac{ip\cdot x}{2}}e^{-\frac{m\omega}{4}x^2}, 
\end{equation}
 \begin{equation}\label{rho_ev}
        \rho_{\psi}(x)=\rho^{2}_{\psi}(x;0)= \left(m\pi \omega\right)^{d/2}e^{-\frac{m\omega}{4}x^2},
      \end{equation}
and 
    \begin{equation}\label{tau_ev}
      \tau_{\psi}(x)=\left(m\pi \omega \right)^{d/2}\left(\frac{d\omega}{4}-\frac{m\omega^2}{8}x^2\right)e^{-\frac{m\omega}{4}x^2}.
    \end{equation}
For $\varphi$, we choose 
\begin{equation}\label{choose:phi}
  \varphi(p)=-\frac{\sqrt{\alpha}\overline{v(p)}}{\epsilon(p)}\left(1+\frac{p\cdot P}{\alpha \mpek \epsilon(p)}\right). 
\end{equation}
In particular, by the definition of $\mpek$ in \eqref{def:MPek},
\begin{equation}\label{PP}
  \int p |\varphi(p)|^2  \dd p =P
\end{equation}
and 
\begin{equation}\label{epsP}
  \int \epsilon(p)|\varphi(p)|^2 dp =\alpha\int \frac{|v(p)|^2}{\epsilon(p)}\dd p+\frac{P^2}{2\alpha\mpek}.
\end{equation} 
Furthermore, for this choice of $\varphi$ we have 
\begin{equation}
  \mathfrak{Re}F(x)=\alpha J(x)+\frac{1}{\alpha}K_P(x)
\end{equation}
where 
\begin{equation}\label{J}
  J(x):=\int \frac{|v(p)|^2}{\epsilon(p)^2}\cos (p\cdot x)\, \dd p
\end{equation}
and 
\begin{equation}\label{K}
  K_P(x):=\frac{1}{\left(\mpek\right)^2}\int \frac{|v(p)|^2}{\epsilon(p)^4}(P\cdot p)^2\cos (p\cdot x) \,\dd p
\end{equation}
as well as 
\begin{equation}
  \mathfrak{Im}F(x)=-P\cdot x-\frac{2}{\mpek}\int \frac{|v(p)|^2}{\epsilon(p)^3}(p\cdot P) \left(\sin (p\cdot x) -p\cdot x\right) \dd p \equiv -P\cdot x+A(x)
\end{equation}
where we used \eqref{PP}. As a consequence, \begin{equation}
\begin{split}
  \GG(R)&=N_\alpha \int e^{-\frac{m}{4}\omega x^2}e^{\mathfrak{Re}F(x)}e^{i((R-P)\cdot x+A(x))}\,\dd x\\ &=N_{\alpha} \int e^{-\frac{m}{4}\omega x^2}e^{\mathfrak{Re}F(x)}\cos\left((R-P)\cdot x+A(x)\right)\, \dd x
  \end{split}
\end{equation} where \begin{equation}
  N_{\alpha}:=\frac{1}{(2\pi)^d}\left(m \pi \omega \right)^{d/2}e^{-F(0)}.
\end{equation}
We finally evaluate
\begin{equation}
  G^{(2)}_{\psi,\varphi}(R,S)=N_{\alpha}\int e^{-\frac{1}{2\omega}(R-S)^2}e^{-\frac{m\omega}{4}x^2}e^{\mathfrak{Re}F(x)}\cos\left(A(x)+(R-P)\cdot x-\frac{(R-S)\cdot x}{2}\right)\dd x.
\end{equation}
In the next step, we perform an asymptotic analysis of the integrals appearing in the definitions of $G,G^{(2)}$ for large values of $\alpha$. 

\bigskip
\underline{\emph{Estimation of the weight integrals}}:
Let 
\begin{equation}\label{def:I}
     I:=N_\alpha \int e^{-\frac{m}{4}\omega x^2}e^{\mathfrak{Re}F(x)}~\dd x;
  \end{equation}
 since $\mathfrak{Re}F(x)\leq \mathfrak{Re}F(0)$, this integral is well-defined. We can hence introduce the probability measure \begin{equation}
   m(x)\dd x =\frac{N_{\alpha }e^{-\frac{m}{4}\omega x^2}e^{\mathfrak{Re}F(x)}}{I}\dd x.
 \end{equation}
Note that all moments of this distribution exist. We denote the expectation value with respect to this distribution by $\langle \cdot \rangle$, which should not be confused with the usual Dirac notation also employed here. The following lemma shows that $m(x)$ is essentially a Gaussian distibution with effective support on a lengthscale $x\sim \alpha^{-1/2}$ dictated by the $F(x)$.
\begin{lemma}\label{lemma_laplace}
For all $r\geq 0$ there exist positive constants $C^{(r)}_1,C^{(r)}_2$ such that for all $\alpha$ large enough and all $P$ with $|P| / \alpha$  small enough we have 
\begin{equation}\label{integral}
  C^{(r)}_1\alpha^{-(r+d)/2}\leq e^{-F(0)}\int |x|^r e^{-\frac{m}{4}\omega x^2}e^{\mathfrak{Re}F(x)}\dd x\leq C^{(r)}_2 \alpha^{-(r+d)/2}.
\end{equation}
\end{lemma}
\begin{proof}
As $\cos x \leq 1-\frac{1}{2}x^2+\frac{1}{24}x^4$, we clearly have, with $J$ defined in \eqref{J}, 
\begin{equation}
  J(x)\leq J(0)-\lambda x^2+\theta |x|^4
\end{equation}
with $\lambda=\frac{1}{2d}\int \frac{p^2|v(p)|^2}{\epsilon(p)^2}\,\dd p$ and $\theta= \frac{1}{24d}\int \frac{|p|^4|v(p)|^2}{\epsilon(p)^2}\,\dd p$. These integrals are finite  by our assumptions on $v$ and $\epsilon$. Moreover the function $K_P$, defined in \eqref{K}, satisfies $K_P(x)\leq K_P(0)$ and hence
 \begin{equation}
  \mathfrak{Re}F(x)\leq F(0)-\alpha \lambda x^2 +\alpha \theta |x|^4 .
\end{equation}
Let us choose  $\varepsilon$ such that $0<\varepsilon<\lambda$, and let $\delta=\sqrt{\frac{\varepsilon}{\theta}}$.  We have that $J(x)<J(0)$ for any $x$ with $|x|>\delta$. By the Riemann--Lebesgue Lemma, $J$ is continuous and vanishes at infinity. It follows that there exists $\xi>0$ such that 
\begin{equation}\label{eta}
  J(x)\leq J(0)-\xi, \quad \forall x: |x|>\delta. 
\end{equation} Note that since $J$ is independent of $\alpha$ and $P$, so are $\delta$ and $\xi$. From \eqref{eta} and from $K_P(x)\leq K_P(0)$ we conclude that
\begin{equation}
  \mathfrak{Re}F(x)\leq F(0)-\alpha\xi \quad \forall x: |x|>\delta. 
\end{equation}
We thus obtain the upper bound
\begin{equation}
\begin{split}
  &\int |x|^r e^{-\frac{m}{4}\omega x^2}e^{\mathfrak{Re}F(x)}\dd x=\int_{|x|\leq \delta} |x|^r e^{-\frac{m}{4}\omega x^2}e^{\mathfrak{Re}F(x)}\dd x+\int_{|x|>\delta} |x|^r e^{-\frac{m}{4}\omega x^2}e^{\mathfrak{Re}F(x)}\dd x  \\ & \leq e^{F(0)}\int_{\mathbb{R}^d} |x|^r e^{-\alpha \left(\lambda-\varepsilon\right)x^2-\frac{m}{4}\omega x^2}\dd x +e^{F(0)-\alpha\xi}\int_{\mathbb{R}^d} |x|^r e^{-\frac{m}{4}\omega x^2}\dd x\\ &= e^{F(0)} \mathcal{C}_r \left(  \left(\alpha(\lambda-\varepsilon)+\frac{m\omega}{4}\right)^{-(r+d)/2}+ e^{-\alpha\xi}{\left(\frac{m\omega}{4}\right)^{-(r+d)/2}}\right)
  \end{split}
\end{equation} 
  where $\mathcal{C}_r=\int_{\mathbb{R}^d} |u|^r e^{-u^2}\dd u=2\pi^{d/2}{\Gamma(\frac{r+d}{2})}/{\Gamma(\frac{d}{2})}$. Since $\omega\sim \sqrt{\alpha}$ and $\xi>0$, the desired upper bound follows. 
  
  For a lower bound we simply use $\cos x\geq 1-\frac{1}{2}x^2$, and consequently 
  \begin{equation}
    \mathfrak{Re}F(x)\geq F(0)-\left(\alpha\lambda+\frac{P^2}{\alpha}\mu\right)x^2
  \end{equation} 
  where 
  \begin{equation}
  \mu=\frac{1}{2d(\mpek)^2}\int |p|^4 \frac{|v(p)|^2}{\epsilon(p)^4}\,\dd p. 
  \end{equation}  
  Thus we can directly bound \begin{equation}
    \int |x|^r e^{-\frac{m}{4}\omega x^2}e^{\mathfrak{Re}F(x)}\dd x \geq\frac{ e^{F(0)}\mathcal{C}_r}{\left(\alpha\lambda+\frac{P^2}{\alpha}\mu+\frac{m}{4}\omega\right)^{\frac{r+d}{2}}}.
  \end{equation} 
  Again, since $\omega\sim \sqrt{\alpha}$, and since $|P|\leq C\alpha$ by assumption, we arrive at the desired conclusion. 
\end{proof}

The lemma implies the bounds 
\begin{equation}\label{moments_estimate}
  \frac{C^{(r)}_1}{C^{(0)}_2}\alpha^{-r/2} \leq \langle |x|^r \rangle \leq \frac{C^{(r)}_2}{C^{(0)}_1}\alpha^{-r/2}. 
\end{equation}
With these preliminary computations and results at hand, we shall now estimate the various terms in \eqref{exp_values}, as well as the norm of $\phi_P$. 
\bigskip

\underline{\textit{Bound on the norm}}: Note that for all $R\in \mathbb{R}^d$ 
\begin{equation}\label{upper_bound_on_G}
  \GG(R)=N_{\alpha} \int e^{-\frac{m}{4}\omega x^2}e^{\mathfrak{Re}F(x)}\cos\left((R-P)\cdot x+A(x)\right)~\dd x\leq I
\end{equation}
with $I$ defined in \eqref{def:I}. Since $\langle \phi_P | \phi_P\rangle =\GG(P)$, $\GG(P)$ is positive. Using $\cos x\geq 1-\frac{1}{2}x^2$ again, we have
\begin{equation}\label{gpp_ri}
  \GG(P)\geq I-\mathcal{J}
\end{equation} 
where 
\begin{equation}
  \mathcal{J}:=\frac{1}{2}N_\alpha \int e^{-\frac{m}{4}\omega x^2}e^{\mathfrak{Re}F(x)}A(x)^2\,\dd x.
\end{equation}Since $|\sin x-x|\leq C|x|^3$ and  $\int \frac{|v(p)|^2}{\epsilon(p)^2}|p|^4\,\dd p$ is finite, we have \begin{equation}\label{A}
  |A(x)|\leq C_A|P||x|^3
\end{equation} for some constant $C_A>0$, independent of $P$ and $\alpha$. Therefore
\begin{equation}\label{R/I}
  \frac{\mathcal{J}}{I}\leq \frac{C^2_AP^2}{2}\langle |x|^6 \rangle \leq C \frac{P^2}{\alpha^3}
\end{equation}
by \eqref{moments_estimate}. 
Hence, if $\alpha$ is large and $|P|\lesssim \alpha$, $\mathcal{J}/I$ is small and we can conclude that 
\begin{equation}\label{102}
  \frac{1}{\langle \phi_P|\phi_P\rangle}=\frac{1}{\GG(P)}\leq \frac{1}{I}\left(1+C\frac{P^2}{\alpha^3}\right)
  \end{equation}
for suitable $C>0$.
This bound on the norm is sufficient for our purpose. 
\vspace{\baselineskip}

\textit{\underline{Bound on the field energy}}: 
Using the definitions, we can express the expected value of the field energy in our trial state as  
\begin{equation}
  \frac{\langle \phi_P|\mathbb{F}|\phi_P\rangle}{\langle \phi_P|\phi_P\rangle} =\int \epsilon(p)|\varphi(p)|^2 \frac{\GG(P-p)}{\GG(P)}~\dd p. 
\end{equation}
Using now \eqref{upper_bound_on_G}, \eqref{gpp_ri} and \eqref{R/I}, we have \begin{equation}
  \frac{\GG(P-p)}{\GG(P)}\leq 1+C\frac{P^2}{\alpha^3},
\end{equation}
 and hence 
 \begin{equation}\label{bound_on_F}
   \frac{\langle \phi_P|\mathbb{F}|\phi_P\rangle}{\langle \phi_P|\phi_P\rangle} 
   \leq \alpha\int \frac{|v(p)|^2}{\epsilon(p)}\dd p+\frac{P^2}{2\alpha\mpek}+C\frac{P^2}{\alpha^2}
\end{equation} 
for $|P|\lesssim \alpha$, where we used \eqref{epsP}. 
\vspace{\baselineskip}

\underline{\textit{Bound on the interaction energy}}:
We have 
\begin{equation}\label{vv}
  \langle \phi_P | \mathbb{V}|\phi_P \rangle = 2 \,\mathfrak{Re}\int v(p)\varphi(p) G^{(2)}_{\psi,\varphi}(P-p,P) \,\dd p. 
\end{equation}
By plugging in \eqref{rho2_ev}, we obtain \begin{equation}\label{VV_final}
  \langle \phi_P | \mathbb{V}|\phi_P \rangle =2 N_{\alpha}\iint  v(p)\varphi(p) e^{-\frac{1}{4m\omega}p^2}e^{-\frac{m\omega}{4}x^2}\cos\left(A(x)-\frac{p\cdot x}{2}\right) e^{\mathfrak{Re}F(x)}\,\dd x \dd p . 
\end{equation}
Let $\tilde{V}$ denote the above expression without the coupling between $p$ and $x$ under the cosine, i.e.,
\begin{equation}
  \tilde{V}= 2 N_{\alpha} \int v(p)\varphi(p) e^{-\frac{1}{4m\omega}p^2}\dd p \int e^{-\frac{m\omega}{4}x^2} e^{\mathfrak{Re}F(x)}\cos A(x)\,\dd x.
\end{equation} 
Using the definition of $\GG$ and plugging in our choice of $\varphi$, we obtain 
\begin{equation}\label{bound_on_Vtilde}
  \tilde{V}=-2 \GG(P)\sqrt{\alpha} \int \frac{|v(p)|^2}{\epsilon(p)}e^{-\frac{p^2}{4m\omega}}\,\dd p.  
\end{equation} 
Note that the contribution  of the $P$-dependent part of $\varphi$ vanishes here by rotation invariance. 
By $e^{-x}\geq 1-x$ and the definition of $\omega$ in \eqref{omega}, this gives 
\begin{equation}
  \frac{\sqrt{\alpha}\tilde{V}}{\langle \phi_P|\phi_P \rangle } \leq -2\alpha \int \frac{|v(p)|^2}{\epsilon(p)}\,\dd p+\frac{d\omega}{4}. 
\end{equation}
We are left with estimating the difference $|\GG(P)^{-1}(\langle \phi_P | \mathbb{V}|\phi_P \rangle -\tilde{V})|$. We apply the elementary inequality  
\begin{equation}
\begin{split}
  & \left|\cos\left(A(x)-\frac{p\cdot x}{2}\right)-\cos A(x)\right|\leq |\cos(A(x))||\cos(p\cdot x/2)-1|+|\sin A(x)||\sin(p\cdot x/2)|  \\ & \leq \frac{(p\cdot x)^2}{8}+|A(x)|\frac{|p||x|}{2}
\end{split}
\end{equation} 
where we used $|\cos z-1|=2|\sin^2z/2|\leq z^2/2$.
Recalling our choice of $\varphi$ in \eqref{choose:phi}, this gives 
\begin{equation}\label{est}
\sqrt{\alpha}  G(P)^{-1} \left(\langle \phi_P | \mathbb{V}|\phi_P \rangle -\tilde{V} \right)\leq  I_a+I_b+II_a+II_b
\end{equation}
with the following terms to estimate:
\begin{equation}\label{Ia}
\begin{split}
  I_a &=\frac{2{\alpha}N_{\alpha}}{G(P)}\iint \frac{|v(p)|^2}{\epsilon(p)} e^{-\frac{1}{4m\omega}p^2}e^{-\frac{m\omega}{4}x^2} e^{\mathfrak{Re}F(x)}\frac{(p\cdot x)^2}{8}\dd x \dd p  \\\ & \leq \frac{\alpha N_{\alpha}}{4I d }\int \frac{p^2|v(p)|^2}{\epsilon(p)} \dd p  \int  x^2  e^{-\frac{m\omega}{4}x^2} e^{\mathfrak{Re}F(x)} \dd x  \left(1+C\frac{P^2}{\alpha^3}\right) =\frac{m\omega^2}{8}\langle x^2 \rangle \left(1+C\frac{P^2}{\alpha^3}\right) 
  \end{split}
\end{equation} where we have used \eqref{102}, the rotation-invariance of $|v|^2/\epsilon$, and the definition of $\omega$ in \eqref{omega}; 
\begin{equation}
\begin{split}
  I_b&=\frac{2{}N_{\alpha}}{G(P)}\iint \frac{ |P\cdot p|~ |v(p)|^2}{ M^{\mathrm{Pek}}\epsilon(p)^2} e^{-\frac{1}{4m\omega}p^2}e^{-\frac{m\omega}{4}x^2} e^{\mathfrak{Re}F(x)}\frac{(p\cdot x)^2}{8}\dd x \dd p  \\ &\leq \frac{|P|~\langle x^2 \rangle }{4\mpek} \left(\int \frac{|p|^3|v(p)|^2}{\epsilon(p)^2}\dd p\right) \left(1+C\frac{P^2}{\alpha^3}\right)\leq C \frac{|P|}{\alpha} 
  \end{split}
\end{equation}
 by \eqref{moments_estimate};
\begin{equation}\begin{split}
  II_a&=\frac{2{\alpha}N_{\alpha}}{G(P)}\iint \frac{|v(p)|^2}{\epsilon(p)}e^{-\frac{1}{4m\omega}p^2}e^{-\frac{m\omega}{4}x^2} e^{\mathfrak{Re}F(x)}|A(x)|\frac{|p||x|}{2}\dd x \dd p  \\ & \leq C_A \alpha |P| ~ \langle |x|^4 \rangle~\left(\int \frac{|p||v(p)|^2}{\epsilon(p)}\,\dd p\right)\left(1+C\frac{P^2}{\alpha^3}\right)\leq C\frac{|P|}{\alpha} 
  \end{split}
\end{equation}
by \eqref{A} and again \eqref{moments_estimate}; finally 
\begin{equation}
\begin{split}
  II_b &= \frac{2 N_{\alpha}}{G(P)}\iint \frac{|P\cdot p| ~|v(p)|^2}{ M^{\mathrm{Pek}}\epsilon(p)^2} e^{-\frac{1}{4m\omega}p^2}e^{-\frac{m\omega}{4}x^2} e^{\mathfrak{Re}F(x)}|A(x)|\frac{|p||x|}{2} \, \dd x \,\dd p  \\ & \leq C_A \frac{P^2}{\mpek} ~\langle |x|^4 \rangle~\left(\int \frac{|v(p)|^2|p|^3}{\epsilon(p)^2}\dd p\right)\left(1+C\frac{P^2}{\alpha^3}\right)  \leq C\frac{P^2}{\alpha^2} .
  \end{split}
\end{equation}
Combining all the estimates, we conclude that  in the regime of large $\alpha$ and small $ |P|/\alpha$ we have 
\begin{equation}\label{V_final}
 \sqrt{\alpha} \frac{\langle \phi_P | \mathbb{V}|\phi_P \rangle}{\langle \phi_P|\phi_P\rangle} \leq -2\alpha\int \frac{|v(p)|^2}{\epsilon(p)}\dd p +\frac{d\omega}{4}+\frac{m\omega^2}{8}\langle x^2 \rangle + C\frac{|P|}{\alpha} .
\end{equation}
\vspace{\baselineskip}
\underline{\textit{Bound on the kinetic energy}}: 
By plugging \eqref{tau_ev} into \eqref{def:T}, we see that the first term in \eqref{exp_values} is given by 
\begin{equation}\label{kin}
\frac 1{2m}    \frac{\langle \phi_P|(P-P_f)^2|\phi_P\rangle}{\langle \phi_P|\phi_P\rangle}=\frac{d\omega}{4}-\frac{m\omega^2}{8}\frac{\langle x^2 \cos A(x)\rangle }{\langle \cos A(x) \rangle },
  \end{equation}
  where $\langle \phi_P|\phi_P\rangle = \GG(P) = I \langle \cos A(x) \rangle$ and, in particular, $0< \langle \cos A(x) \rangle \leq 1$. 
We have, by \eqref{A},
\begin{equation}
   \langle x^2 \cos A(x) \rangle \geq \langle x^2 \rangle -CP^2 \langle |x|^8 \rangle ,
 \end{equation} 
 and thus 
 \begin{equation}
   \frac{\langle x^2 \cos A(x)\rangle }{\langle \cos A(x) \rangle } \geq \langle x^2 \rangle - CP^2 \langle x^8 \rangle \geq \langle x^2 \rangle -CP^2\alpha^{-4} 
 \end{equation}
 using \eqref{moments_estimate}. 
In particular,
 \begin{equation}\label{kin_final}
\frac 1{2m}    \frac{\langle \phi_P|(P-P_f)^2|\phi_P\rangle}{\langle \phi_P|\phi_P\rangle}\leq \frac{d\omega}{4}-\frac{m\omega^2\langle x^2\rangle}{8}+C\frac{P^2}{\alpha^3}.
\end{equation}
Upon adding \eqref{kin_final}, \eqref{V_final}, and \eqref{bound_on_F}, we arrive at the claimed upper bound. \end{proof}

\bigskip
\noindent \textbf{Acknowledgments.} Financial support through the European Research Council (ERC) under the European Union's Horizon 2020 research and innovation programme grant agreement No. 694227 (R.S.) and the Maria Sk\l{}odowska-Curie grant agreement No. 665386 (K.M.)  is gratefully acknowledged.


\begin{thebibliography}{99}
\bibitem{Lan1} L.D. Landau,  \textit{Über die Bewegung der Elektronen in Kristallgitter}, Phys. Z. Sowjetunion. \textbf{3}, 644--645 (1933)
\bibitem{Grusdt} F. Grusdt and E. Demler, \textit{New theoretical approaches to the Bose polarons}, Proc. Int. School of Physics "E. Fermi", Course 191, pp. 235--411 (2016) (available as preprint arXiv:1510.04934 [cond-mat.quant-gas])
\bibitem{Dev} S. Alexandrov and J. Devreese, \textit{Advances in Polaron Physics}, Springer (2010)
\bibitem{Fro1} H. Fr\"ohlich, \textit{Theory of Electrical Breakdown in Ionic Crystals},  Proc. R. Soc. Lond. A \textbf{160}(901), 230--241 (1937)
\bibitem{LT} E.H.Lieb and L.E. Thomas, \textit{Exact Ground-State Energy of the Strong-Coupling Polaron}, Commun. Math. Phys. \textbf{183}, 511--519 (1997); Erratum: Commun. Math. Phys. \textbf{188}, 499--500 (1997)
\bibitem{LY} E.H. Lieb and K. Yamazaki, \textit{Ground-State Energy and Effective Mass of the Polaron}, Phys. Rev. \textbf{111}, 728 (1958)
\bibitem{DV} M.D. Donsker and S.R.S. Varadhan, \textit{Asymptotics for the polaron}, Commun. Pure Appl. Math. \textbf{36}, 505--528 (1983)
\bibitem{Fey1} R.P. Feynman, \textit{Slow Electrons in a Polar Crystal}, Phys. Rev. \textbf{97}, 660 (1955)
\bibitem{Tempere}  W. Casteels, T. Van Cauteren, J. Tempere and J. T. Devreese, \textit{Strong coupling treatment of the polaronic system consisting of an impurity in a condensate}, Laser Physics \textbf{21}, 1480 (2011) 
\bibitem{FrankSe} R. Frank and R. Seiringer, \textit{Quantum corrections to the Pekar asymptotics of a strongly coupled polaron}, Commun. Pure Appl. Math. {\bf 74}, 544--588 (2021)
\bibitem{mys19} K. My\'sliwy, \textit{Ground-state energy of the strongly-coupled polaron in full space -- revisited}, unpublished (2019)
\bibitem{DarioSe} D. Feliciangeli and R. Seiringer \textit{The Strongly Coupled Polaron on the Torus: Quantum Corrections to the Pekar Asymptotics}, Arch. Rat. Mech. Anal. (in press) (2021)
\bibitem{LiSe19} E.H. Lieb and R. Seiringer, \textit{Divergence of the effective mass of a polaron in the strong coupling limit}, J. Stat. Phys. \textbf{180}, 23--33 (2020)
\bibitem{LanPek} L.D. Landau and S.I. Pekar, \textit{Effektivna massa poliarona},  Zh. Eksp. Teor. Fiz. \textbf{18}, 419 (1948)
\bibitem{Pek} S.I. Pekar, \textit{Issledovania po elektronnoi teorii kristallov}, Gostekhizdat, Moskva 1951
\bibitem{LLP} T.D. Lee, F.E. Low, and D. Pines, \textit{The Motion of Slow Electrons in a Polar Crystal}, Phys. Rev. \textbf{90}, 297 (1953)
\bibitem{Dav1} N. Leopold, D. Mitrouskas, S. Rademacher, B. Schlein, and R. Seiringer, \textit{Landau-Pekar equations and quantum fluctuations for the dynamics of a strongly coupled polaron},  arXiv:2005.02098, Pure Appl. Anal. (in press)
\bibitem{nagy} P.Nagy, \emph{A Note to the Translationally-Invariant Strong Coupling Theory of the Polaron}, Czech. J. Phys. B \textbf{39}, 353–356. (1989)
\bibitem{liebse_equi_mass} E.H. Lieb and R. Seiringer,  \textit{Equivalence of two definitions of the effective mass of a polaron}, J. Stat. Phys. \textbf{154}, 51–57 (2014)
\bibitem{Dav2} D.J. Mitrouskas, \textit{A note on the Fröhlich dynamics in the strong coupling limit},  Lett. Math. Phys. \textbf{111}, 45 (2021)
\bibitem{Moeller} J.S. Møller, \textit{The polaron revisited}, Rev. Math. Phys \textbf{18}(5), 485--517 (2006)
\bibitem{Gerlach} B. Gerlach and H. L\"owen, \textit{Analytical properties of polaron systems or: Do polaronic phase transitions exist or not?}, Rev. Mod. Phys. \textbf{63}, 63 (1991)
\bibitem{DeckerPizzo} D. Deckert and A. Pizzo, \textit{Ultraviolet Properties of the Spinless, One-Particle Yukawa Model}, Commun. Math. Phys \textbf{327}, 887--920 (2014)  
\bibitem{Spohn} W. Dybalski and H. Spohn, \textit{Effective mass of the polaron -- revisited},  Ann. Henri Poincare \textbf{21}, 1573--1594 (2020)
\bibitem{LiebLoss} E.H. Lieb and M. Loss, \textit{A Bound on Binding Energies and Mass Renormalization in Models of Quantum Electrodynamics}, J. Stat. Phys. \textbf{108}, 1057–1069 (2002)
\bibitem{PenaArdila} L.A. Peña Ardila, G.E. Astrakharchik, and S. Giorgini, \textit{Strong coupling Bose polarons in a two-dimensional gas}, Phys. Rev. Research \textbf{2}, 023405 (2020)
\bibitem{Timour} T. Ichmoukhamedov and J. Tempere, \textit{Feynman path-integral treatment of the Bose polaron beyond the Fröhlich model}, Phys. Rev. A \textbf{100}, 043605 (2019)
\bibitem{Pastukhov} O. Hryhorchak, G. Panochko, and V. Pastukhov, \textit{Mean-field study of repulsive 2D and 3D Bose polarons}, J. Phys. B: At. Mol. Opt. Phys. \textbf{53}, 205302 (2020)
\bibitem{MiszaEnderalp} E. Yakaboylu, M. Shkolnikov, and M. Lemeshko, \textit{Quantum groups as hidden symmetries of quantum impurities},  Phys. Rev. Lett. \textbf{121}, 255302 (2018)
\end{thebibliography}
\end{document}